\documentclass[a4paper]{llncs}
\usepackage{times}
\usepackage{amssymb}
\usepackage{amsmath}
\usepackage{amsfonts}
\usepackage{graphicx}
\RequirePackage{fancyhdr}
\usepackage{xcolor}
\usepackage{paralist}
\usepackage{boxedminipage}
\usepackage{tikz}
\usetikzlibrary{positioning,arrows,automata,shapes}
\usepackage{url}
\usepackage{subfigure}
\usepackage[linesnumbered, boxed, algosection]{algorithm2e}

\newcommand{\name}[1]{\textsc{#1}}

\newcommand{\arcfunc}{\phi}

\newcommand{\weight}{\rho}

\newenvironment{namedefn}[3]{
\par\addvspace{0.4\baselineskip}\fbox{%
\begin{minipage}[t]{0.9\linewidth}%
\begin{tabular}{p{18mm}p{115mm}}
    \multicolumn{2}{l}{{\name{#1}}} \\
        \textsl{Input:} & {#2} \\ 
        \textsl{Task:} & {#3} \\
   \end{tabular}
\end{minipage}}\par\addvspace{0.4\baselineskip}
}

\title{Parameterized Directed $k$-Chinese Postman Problem and $k$ Arc-Disjoint Cycles Problem on Euler  Digraphs}
\author{Gregory Gutin\inst{1} \and Mark Jones\inst{1} \and Bin Sheng\inst{1} \and Magnus Wahlstr{\"o}m\inst{1}
}
\institute{Royal Holloway, University of London\\
Egham, Surrey TW20 0EX, UK
} 
\begin{document}

\maketitle

\begin{abstract}
\noindent
In the Directed $k$-Chinese Postman Problem ($k$-DCPP), we are given a connected weighted digraph $G$ and asked to find $k$ non-empty closed directed walks covering all arcs of $G$ such that the total weight of the walks is minimum. Gutin, Muciaccia and Yeo (Theor. Comput. Sci. 513 (2013) 124--128) asked for the parameterized complexity of $k$-DCPP when $k$ is the parameter. We prove that the $k$-DCPP is fixed-parameter tractable.

We also consider a related problem of finding $k$ arc-disjoint directed cycles in an Euler digraph, parameterized by $k$. Slivkins (ESA 2003) showed that this problem is W[1]-hard for general digraphs. Generalizing another result by Slivkins, we prove that the problem is fixed-parameter tractable for Euler digraphs. The corresponding problem on vertex-disjoint cycles in Euler digraphs remains W[1]-hard even for Euler digraphs.
\end{abstract}

\section{Introduction}
A digraph $H$ is {\em connected} if the underlying undirected graph of $H$ is connected.
Let $G=(V,A)$ be a connected digraph, where each arc $a \in A$ is assigned a non-negative integer weight $\omega(a)$ ($G$ is a {\em weighted digraph}).  The {\sc Directed Chinese Postman Problem} is a well-studied polynomial-time solvable problem in combinatorial optimization \cite{BaGu2009,EdJo1973,LiZh1988}.


 \begin{namedefn}%
   {{\sc Directed Chinese Postman Problem (DCPP)}}%
   {A connected weighted digraph $G=(V,A).$}%
   {Find a minumum total weight closed directed walk $T$  \newline on $G$  such that every arc of $G$
    is contained in $T$.}%
 \end{namedefn}

\noindent In this paper, we will investigate the following generalisation of DCPP.

\begin{namedefn}%
   {Directed $k$-Chinese Postman Problem ($k$-DCPP)}%
   {A connected weighted digraph $G=(V,A)$ and
   an integer $k$.}%
   {Find a minimum total weight set of  $k$ non-empty \newline closed directed walks such that every
   arc of $G$ is \newline contained  in at least one of them.}%
 \end{namedefn}
Note that the $k$-DCPP can be extended to directed multigraphs (that may include parallel arcs but no loops), but the extended version can be reduced to the one on digraphs by subdividing parallel arcs and adjusting weights appropriately. Since it is more convenient, we consider the $k$-DCPP for digraphs only.

 In the literature, the undirected version of $k$-DCPP, abbreviated $k$-UCPP, has also been studied. 
 If a vertex $v$ of $G$ is part of the input and we require that each of the $k$ walks contains $v$ then
 the $k$-DCPP and $k$-UCPP
 are polynomial-time solvable
 \cite{Zha1992,Per1994}. However,
 in general the $k$-DCCP is NP-complete \cite{GuMuYe2013}, as is the $k$-UCPP  \cite{GuMuYe2013,Tho1997}.

Lately research in parameterized algorithms and complexity\footnote{For terminology and results on parameterized algorithms and complexity we refer the reader to the monographs~\cite{DowneyFellows99,FlumGrohe06,Niedermeier06}.} for the CPP and its generalizations was summarized in \cite{BeNiSoWe} and reported in \cite{Sor2013}.
Several recent results described there are of Niedermeier's group who identified a number of practically useful parameters for the CPP and its generalizations, obtained several interesting results and posed some open problems, see, e.g. \cite{DoMoNiWe2013,SoBeNiWe2011,SoBeNiWe2012}.
van Bevern {\em et al.} \cite{BeNiSoWe} and Sorge \cite{Sor2013} suggested to study the $k$-UCPP as a parameterized problem with parameter $k$ and asked whether the $k$-UCPP is fixed-parameter tractable, i.e. can be solved by an algorithm of running time $O(f(k)n^{O(1)})$, where $f$ is a function of $k$ only and $n= |V|$.

Gutin,  Muciaccia and Yeo \cite{GuMuYe2013} proved that the $k$-UCPP is fixed-parameter tractable. Observing that their approach for the $k$-UCPP is not applicable to the $k$-DCPP, the authors of  \cite{GuMuYe2013} asked for the parameterized complexity of $k$-DCPP parameterized by $k$. In this paper, we show that the $k$-DCPP is also fixed-parameter tractable.

\begin{theorem}\label{thm:main}
The $k$-DCPP is fixed-parameter tractable.
\end{theorem}

Our proof is very different from that in \cite{GuMuYe2013} for the $k$-UCPP. While the latter proof was based on a simple reduction to a polynomial-size kernel, we give a fixed-parameter algorithm directly using significantly more powerful tools. In particular, we use an {\em approximation} algorithm of Grohe and Gr{\"u}ber \cite{GrGr2007} for the problem of finding the maximum number $\nu_0(D)$ of vertex-disjoint directed cycles in a digraph $D$ (this algorithm is based on the celebrated paper by Reed {\em et al.} \cite{ReRoSeTh1996} on bounding $\nu_0(D)$ by a function of $\tau_0(D)$, the minimum size of a feedback vertex set of $D$). We also use
the well-known fixed-parameter algorithm of Chen {\em et al.} \cite{ChLiLuOsRa} for the feedback vertex set problem on digraphs.

We also consider the following well-known problem related to the $k$-DCPP.

\begin{namedefn}%
   {$k$-Arc-Disjoint Cycles Problem ($k$-ADCP)}%
   {A digraph $D$ and
   an integer $k$.}%
   {Decide whether $D$ has $k$ arc-disjoint directed cycles.}%
 \end{namedefn}

 Crucially, we are interested in the $k$-ADCP because given a set of $k$ arc-disjoint cycles, we can solve the $k$-DCPP in polynomial time (see Lemma \ref{lem:kdisjointcycles1}). However, this problem is important in its own right.

The problem is NP-hard in general but polynomial-time solvable for planar digraphs \cite{Luc1976}. In fact, for planar digraphs the maximum number of arc-disjoint directed cycles equals the minimum size of a feedback arc set, see, e.g, \cite{BaGu2009}. It is natural to consider $k$ as the parameter for the $k$-ADCP.
It follows easily from the results of Slivkins \cite{Sli2003} that the $k$-ADCP is W[1]-hard. It remains W[1]-hard for quite restricted classes of directed multigraphs, e.g., for directed multigraphs
which become acyclic after deleting two sets of parallel arcs \cite{Sli2003}.
Here we show that the $k$-ADCP-{\sc Euler}, the
$k$-ADCP on Euler digraphs, is fixed-parameter tractable, generalizing a result in  \cite{Sli2003} (Theorem 4).

\begin{theorem}\label{thm:main2}
The $k$-ADCP-{\sc Euler} is fixed-parameter tractable.
\end{theorem}

Interestingly, the problem of deciding whether a digraph has $k$ vertex-disjoint directed cycles, which is  W[1]-hard (also easily follows from the results of Slivkins \cite{Sli2003}), remains W[1]-hard on Euler digraphs.
Indeed, consider a digraph $D$ and let $\nu_0(D)$ denote the maximum number of vertex-disjoint directed cycles in $D$. Construct a new digraph $H$ from $D$ by adding two new vertices $x$ and $y$, arcs $xy$ and $yx$ and the following extra arcs between $x$ and the vertices of $D$: for each $v\in V(D)$ add $\max \{ d^-(v)-d^+(v),0\}$ parallel arcs $vx$ and $\max\{d^+(v)-d^-(v),0\}$ parallel arcs $xv,$ where $d^-(v)$ and $d^+(v)$ are the in-degree and out-degree of $v$, respectively. To eliminate parallel arcs, it remains to subdivide all arcs between $x$ and $V(D)$. Now it is sufficient to observe that $H$ is Euler and $\nu_0(H)=\nu_0(D)+1$. 

To prove Theorems \ref{thm:main} and  \ref{thm:main2} we study the following problem that generalizes the $k$-DCPP (in the case when an optimal solution
exists in which the number of times each arc is visited by every closed
walk is restricted)
and $k$-ADCP.
Let $b \le c$ be non-negative integers.

\begin{namedefn}%
   {Directed $k$-Walk $[b,c]$-Covering Problem ($k[b,c]$-DWCP)}%
   {A connected weighted digraph $G=(V,A)$ and\newline
   an integer $k$.}%
   {Find a minimum total weight set of $k$ non-empty \newline closed directed walks in which every
   arc of $G$ appears \newline between $b$ and $c$ times.}%
 \end{namedefn}
 
 Let $D$ be a digraph. For a vertex ordering $\nu = (v_1, v_2, \dots, v_n)$ of $V(D)$, the {\em cutwidth} of $\nu$ is the maximum number of arcs between $\{1, \dots, i\}$ and $\{i+1, \dots n\}$ over all $i \in [n]$.
The {\em cutwidth} of $D$ is the minimum cutwidth of all vertex orderings of $V(D)$.

 In Section \ref{sec:dpa} we will prove the following theorem.

 \begin{theorem}\label{thm:dpa}
   Let $(G,k)$ be an instance of $k[b,c]$-DWCP and suppose we are given a vertex ordering $\nu =~(v_{1},v_{2},\dots,v_{n})$ of $G$ with cutwidth at most $p$.
  Then $(G,k)$ can be solved in time  $O^*((c2^k)^p4^k)$.
 \end{theorem}

Note that when
$c$ and $p$
are upper-bounded by functions of $k$,  the algorithm of this theorem is fixed-parameter.


In order to apply Theorem \ref{thm:dpa} to the $k$-DCPP and $k$-ADCP-{\sc Euler}, we first need to find a vertex ordering of bounded cutwidth.
This is done using Lemma \ref{lem:ptww}, which given an Euler directed graph, either finds a vertex ordering with cutwidth bounded by a function of $k$, or finds $k$ arc-disjoint cycles. (For the $k$-DCPP, we apply Lemma \ref{lem:ptww} to
an Euler directed multigraph
derived from a solution to the DCPP on $G$.) If $k$ arc-disjoint cycles are found, then the $k$-ADCP-{\sc Euler} is solved. In the case of the $k$-DCPP, it remains to use Lemma \ref{lem:kdisjointcycles1}, which shows that given $k$ arc-disjoint cycles
(in the derived directed multigraph),
we can solve the $k$-DCPP on $G$ in polynomial time.

If we find a vertex ordering of cutwidth $p(k)$, we can solve the $k$-ADCP-{\sc Euler}  by applying Theorem \ref{thm:dpa} with $b = 0, c = 1$.
In the case of the $k$-DCPP, $b = 1$ and it remains to find an upper bound on $c$. This is done using Lemma \ref{lem:multiplicity} proved in Section \ref{sec:struct}, which
shows that if an optimal solution of DCPP traverses
each arc less than $k$ times then there is an optimal solution for the
$k$-DCPP such that no arc is visited more than $k$ times in total by the
$k$ walks of the solution.
If an optimal solution of DCPP visits an arc at least $k$ times, then the derived graph for this solution contains at least $k$ arc-disjoint cycles and again we may use  Lemma \ref{lem:kdisjointcycles1}.
Thus, starting from an arbitrary optimal solution of DCPP, we may either apply Theorem~\ref{thm:dpa} with $c=k$, or Lemma~\ref{lem:kdisjointcycles1}.

The paper is organised as follows. In Section \ref{sec:struct}, we prove six lemmas providing structural results for the $k$-DCPP and $k$-ADCP-{\sc Euler}. In Sections \ref{sec:dpa} and \ref{sec:main}, we prove Theorem \ref{thm:dpa} and the main two results of the paper, Theorems \ref{thm:main} and  \ref{thm:main2}. We conclude the paper with brief discussions of open problems in Section~\ref{sec:dis}.

In what follows, all walks and cycles in directed multigraphs are directed. For a positive integer $p$, $[p]$ will denote the set $\{1,2,\ldots, p\}$.
For integers $a \le b$, $[a,b]$ will denote the set $\{a, a+1, \ldots, b\}$.
Given a directed graph $D$, a \emph{feedback vertex set} for $D$ is a set $S$ of vertices such that $D - S$ contains no directed cycles.
A \emph{feedback arc set} for $D$ is a set $F$ of arcs such that $D - F$ contains no directed cycles.
 A vertex $v$ of a digraph is {\em balanced} if the
in-degree of $v$ equals its out-degree.

\section{Structural Results and Fixed-Parameter Algorithms}\label{sec:struct}

Recall that a directed multigraph $H$ is Euler (i.e., has an Euler trail) if and only if $H$ is connected and
every vertex of $H$ is balanced
\cite{BaGu2009}.

The next lemma is a simple sufficient condition for an Euler digraph to contain $k$ arc-disjoint cycles.

\begin{lemma}\label{lem:Edeg}
Every Euler digraph $D$ having a vertex of out-degree at least $k\ge 1$,
contains $k$ arc-disjoint cycles that can be found in polynomial time.
\end{lemma}
\begin{proof}
For $k=1$, it is true as $D$ has a cycle that can be found in polynomial time. Let $k\ge 2$ and let $C$ be a cycle in $D$. Observe that after deleting the arcs of $C$, $D$ has a vertex of out-degree at least $k-1$ and we are done by induction hypothesis.\qed
\end{proof}

Reed {\em et al.}  \cite{ReRoSeTh1996} proved that there is a function $f:\ \mathbb{N}\rightarrow \mathbb{N}$ such that for every $k$, if a digraph $D$ does not have $k$ arc-disjoint cycles, then it has a feedback arc set with at most $f(k)$ arcs. The celebrated result of Reed {\em et al.}  \cite{ReRoSeTh1996} can be easily extended to directed multigraphs by subdividing parallel arcs.
Using this result, Grohe and Gr{\"u}ber \cite{GrGr2007} showed that there is a non-decreasing and unbounded function $h:\ \mathbb{N}\rightarrow \mathbb{N}$ and
a fixed-parameter algorithm that for a digraph $D$
returns at least $h(k)$ arc-disjoint cycles if $D$ has at least $k$ arc-disjoint cycles.

Let $h^{-1}:\ \mathbb{N}\rightarrow \mathbb{N}$ be defined by $h^{-1}(q)=\min\{p:\ h(p)\ge q\}.$ Since $h$ is a non-decreasing and unbounded function, $h^{-1}$ is a non-decreasing and unbounded function.
Combining the above results, we find that for every digraph $D$, either the algorithm of Grohe and Gr\"uber returns at least $k$ arc-disjoint cycles, or $D$ has a feedback arc set of size at most $f(h^{-1}(k))$.

Chen {\em et al.} \cite{ChLiLuOsRa} designed a fixed-parameter algorithm that decides whether a digraph $D$ contains a feedback vertex set of size $k$ ($k$ is the parameter).
As this is an iterative compression algorithm, it can be easily modified to an algorithm  for finding a minimum feedback vertex set in $D$ (the running time of the latter algorithm is $q(\tau_0(D))n^{O(1)},$ where $\tau_0(D)$ is the minimum size of a feedback vertex set in $D$, $n=|V(D)|$ and $q(k)=4^kk!$).
The modified algorithm can be used for finding a minimum feedback arc set in $D$ as $D$ can be transformed, in polynomial time, into another digraph $H$ such that $D$ has a feedback arc set of size $k$ if and only if $H$ has a feedback vertex set of size $k$, see, e.g., \cite{BaGu2009} (Proposition 15.3.1).

\begin{lemma}\label{lem:cyclesORfas}
There is a function $g:\ \mathbb{N}\rightarrow \mathbb{N}$ 
and a fixed-parameter algorithm 
such that for a digraph $D$, the algorithm returns either $k$ arc-disjoint cycles or a feedback arc set of size at most $g(k)$.
\end{lemma}
\begin{proof}
%
%
%
%
%
%


%
%
Run the Grohe-Gr{\"u}ber algorithm on $D$. Either the algorithm returns at least $k$ arc-disjoint cycles, or
we know that $D$ has no $h^{-1}(k)$
arc-disjoint cycles and so by the result of Reed {\em et al.} $D$ has a feedback arc set of size at most $f(h^{-1}(k)).$ We can use the algorithm of Chen {\em et al.} to find in $D$ a minimum feedback arc set. We may set $g(k)=f(h^{-1}(k))$.\qed
\end{proof}

\begin{lemma}\label{lem:ptww}
Let $g:\ \mathbb{N}\rightarrow \mathbb{N}$ be the function in Lemma \ref{lem:cyclesORfas}.
Let $D$ be an Euler directed multigraph. We can obtain either $k$ arc-disjoint cycles of $D$ or a vertex ordering of cutwidth at most $2g(k)$.
\end{lemma}
\begin{proof}
Let us run the procedure of Lemma \ref{lem:cyclesORfas} for $D$ and $k$. If we get $k$ arc-disjoint cycles, we are done. Otherwise,
we get a feedback arc set $F$ of $D$ such that $|F|\le g(k)$. Then $D' = D - F$ is  an acyclic digraph.
We let $\nu=(v_1,\ldots,v_n)$ be an acyclic ordering of $D'$, i.e., $D'$ has no arc of the form $v_iv_j$, $i>j$,
(it is well-known that such an ordering exists \cite{BaGu2009}).
Now $\nu$ is a vertex ordering for $D$ with at most $|F|$ arcs from $\{v_{i+1},\dots,v_{n}\}$ to  $\{v_{1},\dots,v_{i}\}$ for each $i\in[n-1]$,
and because $D$ is Euler there are the same number of arcs from $\{v_{1},\dots,v_{i}\}$ to $\{v_{i+1},\dots,v_{n}\}$ \cite[Corollary 1.7.3]{BaGu2009}.
So $\nu$ is a vertex ordering with cutwidth at most $2g(k)$.\qed
\end{proof}

In the rest of this section, $G=(V,A)$ is a connected weighted directed graph. For a solution $T=\{T_1,\dots,T_k\}$ to the $k$-DCPP on $G$ ($k\ge 1$), let $G_T=(V,A_T)$, where $A_T$ is a multiset containing all arcs of $A$, each as many times as it is traversed in total by $T_1\cup\dots\cup T_k$.

Lemmas \ref{thm:CPPleqkCPP} and \ref{lem:kdisjointcycles1} are similar to two simple results obtained for the $k$-UCPP in \cite{GuMuYe2013}.
Note that given $k$ closed walks which cover all the arcs of a digraph, their union is a closed walk covering all the arcs and, therefore, it is a solution for the DCPP. Hence, the following proposition holds.

\begin{lemma}\label{thm:CPPleqkCPP}
 The weight of an optimal solution for the $k$-DCPP on $G$ is not smaller than the weight of an optimal solution for the DCPP on $G.$
\end{lemma}

\begin{lemma}\label{lem:kdisjointcycles1}
Let $T$ be an optimal solution for the DCPP on $G$.
If $G_T$ contains at least $k$ arc-disjoint cycles, then the weight of an optimal solution for the $k$-DCPP on $G$ is equal to the weight of an optimal solution of the DCPP on $G$.
Furthermore if $k$ arc-disjoint cycles in $G_T$ are given, then an optimal solution for the $k$-DCPP can be found in polynomial time.
\end{lemma}
\begin{proof}
Note that $G_T$ is an Euler directed multigraph and so every vertex of $G_T$
is balanced.
Let $\cal C$ be any collection of $k$ arc-disjoint cycles in $G_T$. Delete all arcs of $\cal C$ from $G_T$ and observe that every vertex in the remaining directed multigraph $G'$
is balanced.
Find an optimal DCPP solution for every connected component of $G'$ and append each such solution $F$ to a cycle in $\cal C$ which has a common vertex with $F$.  As a result, in polynomial time, we obtain a collection $Q$ of $k$ closed walks for the $k$-DCPP on $G$ of the same weight as $T$. So $Q$ is optimal by Lemma \ref{thm:CPPleqkCPP}.\qed
\end{proof}

For a directed multigraph $D$, let $\mu_D(xy)$ denote the multiplicity of an arc $xy$ of $D$. The {\em multiplicity} $\mu(D)$ of $D$ is the maximum of the multiplicities of its arcs.
Thus, Lemma \ref{lem:kdisjointcycles1} implies that if $\mu(G_T) \ge k$
for any optimal solution $T$ of the DCPP on $G$, then there is an optimal
solution of the $k$-DCPP on $G$ with weight equal to the weight of $G_T$.
The next lemma helps us in the case that $\mu(G_T) \le k-1$.

\begin{lemma}\label{lem:multiplicity}
Let $T$ be an optimal solution of the DCPP on $G$ such that $\mu(G_T)\le k-1$. Then there is an optimal solution $W$ for the $k$-DCPP on $G$ such that $\mu(G_W)\le k$.
\end{lemma}
\begin{proof}
Let $T$ be  an optimal solution of DCPP on $G$ and let $\mu(G_T)\le k-1$. Suppose that there is an optimal solution $W$ of the $k$-DCPP on $G$ such that  $\mu(G_W)>k.$

Let $\rho(xy)=\mu_{G_W}(xy)-\mu_{G_T}(xy)$ for each arc $xy$ of $G$.
Consider a directed multigraph $H'$ with the same vertex set as $G$ and in which $xy$ is an arc of multiplicity $|\rho(xy)|$ if it is an arc in $G$ and $\rho(xy)\neq 0$. We say that an arc $xy$ of $H'$ is  {\em positive} ({\em negative}) if $\rho(xy)>0$ ($\rho(xy)<0$). Now reverse every negative arc of $H'$ (i.e., replace every negative arc $uv$ by the negative arc $vu$) keeping the weight of the arcs the same. We denote the resulting directed multigraph by $H.$

For a digraph $D$ and its vertex $x$, let $N^+_D(x)$ and  $N^-_D(x)$ denote the sets of out-neighbors and in-neighbors of $x$, respectively.
Since $G_T$ and $G_W$ are both Euler directed multigraphs, we have that
$$\sum_{y\in N_{H'}^+(x)}\rho(xy) = \sum_{z\in N_{H'}^-(x)}\rho(zx) \mbox{ implying } \sum_{u\in N_{H}^+(x)}\mu(xu) = \sum_{v\in N_{H}^-(x)}\mu(vx)$$
for each vertex $x$ in $G$. So, every vertex in $H$ has the same in-degree as out-degree.
Thus, the arcs of $H$ can be decomposed into a collection ${\cal C}=\{C_1,\ldots ,C_t\}$ of cycles.  We define the weight $\omega(C_i)$ of a cycle $C_i$ of $\cal C$ as the sum of the weights of its positive arcs minus the sum of the weights of its negative arcs, and assume that $\omega(C_1)\le \cdots \le \omega(C_t)$.

Set $F_0=G_T$ and for $i\in [t]$,  construct $F_{i}$ from $F_{i-1}$ as follows: for each arc $xy$ of $C_{i}$, if $xy$ is a positive arc in $H$ add a copy of $xy$ to $F_{i-1}$ and if  $xy$ is a negative arc in $H$ remove a copy of $yx$ from $F_{i-1}$. Since for each arc $uv$ of $G$, $\mu_{G_T}(uv)\ge 1$ and $\mu_{G_W}(uv)\ge 1$, we have $\mu_{F_{i}}(uv)\ge 1$.
Each vertex of $F_{i}$ is balanced, so
$F_{i}$ is a solution of DCPP on $G$. Since $T$ is optimal, $\omega(F_0)\le \omega(F_1)=\omega(F_0)+\omega(C_1)$ and so $\omega(C_1)\ge 0$.
Due to the ordering of cycles of $\cal C$ according to their weights, $\omega(C_i)\ge 0$ for $i\in [t]$.
Thus, $\omega(F_{i})\ge \omega(F_{i-1})$ for $i\in [t]$.

Since  $\mu(F_0)\le k-1$ and $\mu(F_t)>k$, there is an index $j$ such that $\mu(F_j)=k$. Then the out-degree of some vertex of $F_j$ is at least $k$ and so by Lemma \ref{lem:Edeg}, $F_j$ has $k$ arc-disjoint cycles. Similarly to Lemma \ref{lem:kdisjointcycles1}, it is not hard to show that there is a solution $U$ of $k$-DCPP on $G$ of weight $\omega(F_j)$. Since $W$ is optimal and
$\omega(F_j)\le \omega(F_t) = \omega(G_W)$, $U$ is also optimal and we are done.\qed
\end{proof}

\section{Proof of Theorem \ref{thm:dpa}}\label{sec:dpa}

Theorem \ref{thm:dpa} is proved by providing a dynamic programming (DP) algorithm of required complexity. We first make an observation to simplify the DP algorithm.

\begin{lemma}
\label{lem:dpjustcycles}
Let $G=(V,A)$ and $k$ define an instance of $k[b,c]$-DWCP. 
The instance is positive and and the weight of an optimal solution is $\weight$ if and only if there are (not necessarily connected) non-empty directed multigraphs $G_1, \ldots, G_k$  with the following properties:
\begin{itemize}
\item All multigraphs $G_1, \ldots, G_k$ use only arcs of $G$ (each, possibly, multiple number of times);
\item $G_1$ is a balanced multigraph;
\item For $2 \leq i \leq k$, $G_i$ is a balanced digraph (with no parallel arcs);
\item Each arc $a \in A$ occurs between $b$ and $c$ times in the multigraph\footnote{Here, as in the proof, the union of multigraphs means that the multiplicity of an arc in the union equals the sum of multiplicities of this arc in the multigraphs of the union.}  $G_1\cup \dots \cup G_k$, and the total weight of this multigraph is $\weight$.
\end{itemize}
\end{lemma}
\begin{proof}
On the one hand, let $W_1, \ldots, W_k$ be a solution to the $k[b,c]$-DWCP instance, where each $W_i$ is a closed directed walk. For each $i\in [k]$, let $Q_i$ be the directed multigraph whose vertices are the vertices visited by $W_i$ and which contains an arc $uv$ of multiplicity $\mu$ if $uv$ is traversed exactly $\mu$ times by $W_i$.
For each $i \geq 2$, if $Q_i$ has parallel arcs, let $G_i$ be a cycle in $Q_i$ and let $Q'_i=Q_i\setminus A(G_i)$ and, otherwise (i.e., $Q_i$ has no parallel arcs), let $G_i=Q_i$ and let $Q'_i$ be empty.
Now let $G_1=Q_1\cup Q_2'\cup \dots \cup Q_k'$. Observe that all properties of the lemma are satisfied.

On the other hand, consider directed multigraphs $G_1, \ldots, G_k$ satisfying the properties of the lemma. If all multigraphs $G_i$ are connected, then we are done.
If $b=0$, then we may replace each graph $G_i$ with a cycle $C_i$ contained in $G_i$, and produce a solution to $k[b,c]$-DWCP that consists of $k$ (not necessarily pairwise arc-disjoint) cycles.

Finally, if not all multigraphs are connected and $b>0$, we proceed as follows. First, select for each multigraph $G_i$, $i>1$ an arbitrary connected component $H_i$, and move all other components of $G_i$ to $G_1$, increasing arc multiplicity as appropriate.
Next, as long as $G_1$ remains unconnected, let $H$ be an arbitrary connected component of $G_1$. As $b>0$ and $G$ is connected, some component $H_i$, $i>1$ must intersect a vertex of $H$; we may move $H$ to the multigraph $G_i$ and maintain that $G_i$ is connected.
Repeat this until $G_1$ (and hence each multigraph $G_i$) is connected. Note that this does not change the arc multiplicity or the weight of the solution. Now every multigraph $G_i$ for $i \in [k]$ is balanced and connected, i.e., Euler, and we can find an Euler tour $W_i$ for each graph $G_i$, which forms the solution to the $k[b,c]$-DWCP instance.\qed
\end{proof}

Let $\nu =~(v_{1},v_{2},\dots,v_{n})$ be a vertex ordering of a digraph $G$ of cutwidth at most $p$.
For each $i \in \{0, 1, \dots, n\}$, let $E_i$ be the set of arcs of the form $v_jv_h$ or $v_hv_j$, where $j \le i$ and $h > i$. Note that in particular $E_0 = \emptyset$ and $E_n = \emptyset$. As $\nu$ has cutwidth at most $p$, $|E_i|\le p$ for each $i$.
We refer to $E_0, E_1, \dots, E_n$ as the \emph{arc bags} of $\nu$.
For each $i \in \{0, 1, \dots, n\}$, let $\gamma(i) = \bigcup_{0\leq j\leq i} E_{j}$.
For a vertex $ v \in V$, let $A^+(v) = \{vu \in A:\ u \in V\}$ and $A^-(v) = \{uv \in A:\ u\in V\}$.

We now give an intuitive description of the DP algorithm before giving technical details.
Our DP algorithm will process each arc bag of $\nu$ in turn, from $E_{0}$ to $E_{n}$.
For each arc bag $E_{i}$, we store the weights of a range of partial solutions. A partial solution consists of a multiset $A_1$ and sets $A_2, \ldots, A_k$ of arcs in $\gamma(i)$. Each $A_j$ is to be thought of the (multi)set of arcs in $G_j$ (defined in Lemma \ref{lem:dpjustcycles}) taken from  $\gamma(i)$.
A function $\arcfunc$ is used to represent how many times each arc in the bag $E_{i}$ is used by each (multi)set $A_j$ in the solution. Finally, a set $S$ tracks which (multi)sets are non-empty. This is to ensure we don't produce a solution which uses less than $k$ non-empty walks.
For each arc bag $E_{i}$, and every choice of $\arcfunc, S$ respecting the conditions of Lemma~\ref{lem:dpjustcycles},
we will calculate the minimum weight of a partial solution corresponding to these choices.

Let us make these notions more precise.
Let $E_{i}$ be an arc bag in $\nu$, and let $\arcfunc$ be a function $E_{i} \times [k] \rightarrow [0,c]$
such that for each $a \in E_i$ we have $\sum_j \arcfunc(a,j) \in [b,c]$ and $\arcfunc(a,j)\leq 1$ for $2 \leq j \leq k$.
Let $S$ be a subset of $[k]$. For a vertex $v$ and multiset $M$ of arcs, let $A^+(v,M)$ be the multiset of arcs from $M$ leaving $v$, and similarly let $A^-(v,M)$ be the multiset of arcs from $M$ entering $v$.
Then we define $\chi(E_{i},\arcfunc,S)$ to be the minimum integer $\weight$
for which there exist arc multisets $A_1, \dots, A_k$ satisfying the following conditions:

\begin{enumerate}
 \item\label{con:ffunction} For every arc $a \in E_{i} $ and every $j \in [k]$, $A_j$ contains exactly $\arcfunc(a,j)$ copies of $a$;
 \item\label{con:edgecovering} For every arc $a \in  \gamma(i)$, the multiset $A_1 \cup \dots \cup A_k$ contains between $b$ and $c$ copies of $a$;
 \item\label{con:vertexbalance} For every $h \le i$  and every $j \in S$, $|A^+(v_h, A_j)| = |A^-(v_h, A_j)|$;
 \item\label{con:sset} For every $j \in [k]$, $A_j \neq \emptyset$ if and only if $j \in S$;
 and
 \item\label{con:wweight} $\sum_{j \in [k]} \sum_{a \in A_j} \omega(a) = \weight$.
\end{enumerate}

Note that $|A^+(v_h,A_j)|$ and  $|A^-(v_h,A_j)|$ are the numbers of arcs in $A_j$
leaving and entering $v_h$, respectively, and that the second sum in
Condition \ref{con:wweight} is taken
over all arcs in multiset $A_j$, i.e., over every copy of an arc in $A_j$.

 If no such integer $\weight$ exists, then we let $\chi(E_{i},\arcfunc,S) = \infty$.

 Observe that if $E_i, \arcfunc, S$ and $\weight$ together with arc multisets $(A_1, \dots, A_k)$ satisfy the above conditions, then $\chi(E_{i},\arcfunc,S) \le \weight$. In such a case we will call $(A_1, \dots , A_k)$ a \emph{witness} for $\chi(E_{i},\arcfunc,S) \le \weight$. Thus, $\chi(E_{i},\arcfunc,S)$ is the minimum $\weight$ such that there exists a witness for  $\chi(E_{i},\arcfunc,S) \le \weight$.
%

 The next lemma shows that we can solve the $k[b,c]$-DCPP by finding the values $\chi(E_{i},\arcfunc,S)$.
  Since $E_n=\emptyset$, the only function $\phi:\
E_n\times [k] \rightarrow [b,c]$ is the empty function.

 \begin{lemma}\label{lem:root}
 Let $\phi:\ E_n\times [k] \rightarrow
[b,c]$ be the empty function.
Then  $\chi(E_n, \phi, [k]) = \infty$ if there is no solution for the $k[b,c]$-DCPP on $G$, and otherwise
  $\chi(E_n, \phi, [k])$ is the minimum total weight of a solution for $k[b,c]$-DCPP.
 \end{lemma}
 \begin{proof}
We will show that (a) if $\chi(E_n, \phi, [k]) = \weight \neq \infty$, then there
  exists a solution for the $k[b,c]$-DCPP on $G$ with weight $\weight$; and  that (b) if there exists a solution for the $k[b,c]$-DCPP on $G$ with weight $\weight$, then there exists a witness for $\chi(E_n, \phi, [k]) \le \weight$.


  In what follows it will be useful to observe that $\gamma(n) = A(G)$.

  Suppose  that $\chi(E_n, \phi, [k]) = \weight \neq \infty$ and $(A_1, \dots, A_k)$ is a witness for $\chi(E_n, \phi, [k]) \le \weight$.
  By Condition \ref{con:vertexbalance} of $\chi(E_n, \phi, [k])$, every vertex is balanced with respect to each arc (multi)set $A_j$,
  and by Condition \ref{con:sset}, each $A_j$ is non-empty.
  Thus, $A_1$ forms the arc (multi)set of a balanced directed multigraph and $A_j$, $j>1$ the arc set of a balanced digraph,
  and by Condition \ref{con:edgecovering}, every arc in $G$ appears between $b$ and $c$ times in these (multi)sets.
  By Lemma~\ref{lem:dpjustcycles}, the arcs of the multiset $A_1\cup \ldots\cup A_k$ can be partitioned into a solution for the $k[b,c]$-DCPP,
  which by Condition \ref{con:wweight} and minimality of $\weight$ has total weight exactly $\weight$.
  Thus there exists a solution for the $k[b,c]$-DCPP on $G$ with weight $\weight$.

  Now suppose that there exists a solution for the $k[b,c]$-DCPP on $G$ with weight $\weight$;
  by Lemma~\ref{lem:dpjustcycles}, there then exist non-empty balanced directed multigraphs $G_1, \ldots, G_k$ of total weight $\weight$, where every arc appears between $b$ and $c$ times in total,
  and where $G_j$ for $j>1$ has no parallel arcs. Letting $A_j$ be the arc (multi)set of $G_j$ for each $j \in [k]$, we find that
  $(A_1, \dots, A_k)$ is a witness for  $\chi(E_n, \phi, [k]) \le \weight$.
  As $E_n  = \emptyset$, Condition \ref{con:ffunction} of $\chi(E_n, \phi, [k])$ is trivially satisfied.
  Condition \ref{con:edgecovering} is satisfied by the conditions in Lemma~\ref{lem:dpjustcycles}.
   Since every vertex in a balanced directed multigraph is balanced,
Condition \ref{con:vertexbalance} is satisfied.
  As each of the $k$ multigraphs is non-empty, Condition \ref{con:sset} is satisfied.
  Finally, as the multigraphs have total weight $\weight$, Condition \ref{con:wweight} is satisfied.
  Thus $(A_1, \dots, A_k)$ is a witness for  $\chi(E_n, \phi, [k]) \le \weight$, as required.\qed
\end{proof}

Due to the space limit, we place the proof of the next lemma in the Appendix.

 \begin{lemma}\label{lem:parent}
  Consider an arc bag  $E_{i}$, for $i \ge 1$.
   Let $E_i^* = E_{i} \setminus E_{i-1}$.
   For any  $\arcfunc :E_{i} \times [k] \rightarrow [0,c]$ and $S\subseteq [k]$,
   let $Y = \sum_{j\in S}\sum_{a \in E_i^*}\arcfunc(a,j) \cdot \omega(a)$.

   If there exists $a \in E_i$ such that  $\sum_{j \in [k]}\arcfunc(a,j) < b$ or  $\sum_{j \in [k]}\arcfunc(a,j)>c$, then
   $\chi(E_{i},\arcfunc,S) = \infty$.

  Otherwise, the following recursion holds:

  \[\chi(E_{i},\arcfunc,S) =  Y + \min_{\arcfunc',S'}\chi(E_{i-1},\arcfunc',S')\]

   where the minimum is taken over all $\arcfunc' :E_{i-1} \times [k] \rightarrow [0,c]$, and $S'\subseteq [k]$ satisfying the following conditions:
   \begin{itemize}
    \item For all $a \in E_{i}\cap E_{i-1}$ and all $j \in [k]$, $\arcfunc'(a,j) = \arcfunc(a,j)$;
   \item For all $j \in [k]$,
   \begin{eqnarray*}
      \sum_{a \in A^+(v_i) \cap E_{i-1}}\phi'(a,j) + \sum_{a \in A^+(v_i) \cap E_i} \phi(a,j) \\
    =   \sum_{a \in A^-(v_i) \cap E_{i-1}}\phi'(a,j) + \sum_{a \in A^-(v_i)\cap E_i}\phi(a,j).
   \end{eqnarray*}

    \item $S = S' \cup \{j \in [k]: \sum_{a \in E_i^* }\arcfunc(a,j) > 0\}$.
   \end{itemize}
    If there are no $\arcfunc',S'$ satisfying these conditions, then $\chi(E_{i},\arcfunc,S) = \infty$.

 Furthermore, if there exist $\arcfunc',S'$ satisfying the above conditions and we are given a witness $(A_1', \dots, A_k')$ for $\chi(E_{i-1}, \arcfunc', S') \le \rho'$, then we can construct a witness for $\chi(E_{i},\arcfunc,S) \le Y + \rho'$ in polynomial time.
 \end{lemma}



%
%
%
We are now ready to prove Theorem \ref{thm:dpa}. 



 \noindent{\bf Theorem}  \ref{thm:dpa}
 {\em  Let $(G,k)$ be an instance of $k[b,c]$-DWCP and suppose we are given a vertex ordering $\nu =~(v_{1},v_{2},\dots,v_{n})$ of $G$ with cutwidth at most $p$.
  Then $(G,k)$ can be solved in time  $O^*((c2^k)^p4^k)$.}

\begin{proof}
Our DP algorithm calculates all values $\chi(E_{i},\arcfunc,S)$ with $\arcfunc(\cdot,j)\leq 1$ for $j>1$ in a bottom-up manner, that is, we only calculate values $\chi(E_{i}, \cdot, \cdot)$ after all values $\chi(E_{j}, \cdot, \cdot)$ have been calculated for  $0\le j<i$ (we use the recursion of Lemma
\ref{lem:parent}).

 Each arc bag $E_{i}$ of $\nu$ contains at most $p$ arcs. For each arc $a$, there are $c+1$ options for $\arcfunc(a,1)$ and $2$ options for $\arcfunc(a,j)$ for each $j>1$, i.e., $(c+1)2^{k-1} \leq c2^k$ options per arc.
 Thus there are at most $(c2^k)^p$ valid choices for $\arcfunc :E_{i} \times [k] \rightarrow [0,c]$.
 As there are $2^k$ choices for a set $S \subseteq [k]$, the total size of each DP table is $O((c2^k)^p2^k)$.



Since $E_0=\emptyset$, the only function $\phi:\
E_0\times [k] \rightarrow [0,c]$ is the empty function.
It is easy to see that $\chi(E_0, \phi, S) = 0$ if $S = \emptyset$, and $\infty$ otherwise.
To speed up the application of Lemma~\ref{lem:parent} for $E_i$, $1 \leq i \leq n$, we form an intermediate table $T$ from the data for bag $E_{i-1}$.
Call two entries $\chi(E_i, \phi, S)$ and $\chi(E_{i-1}, \phi', S')$ \emph{compatible} when the conditions in Lemma~\ref{lem:parent} are met (i.e., $\chi(E_{i-1}, \phi', S')$ is one of the entries included in the minimisation for $\chi(E_i, \phi, S)$).
Let the \emph{signature} of entry $\chi(E_{i-1}, \phi', S')$ be $(\phi'', d_1,\ldots, d_k, S')$, where $\phi''$ is $\phi'$ restricted to arcs $E_{i-1} \cap E_i$, and where $d_j=\sum_{a \in A^+(v_i) \cap E_{i-1}}\phi'(a,j)-\sum_{a \in A^-(v_i) \cap E_{i-1}} \phi'(a,j)$
is the imbalance at $v_i$ in walk number $j$.
Observe that whether an entry $\chi(E_{i-1}, \phi', S')$ is compatible with the entry $\chi(E_i,\phi,S)$ can be determined from the signature alone.
Thus, for every signature $(\phi'',d_1,\ldots,d_k,S')$ we let $T(\phi'',d_1,\ldots,d_k,S')$ contain the minimum value over all entries $\chi(E_{i-1},\ldots)$ with matching signature; this can be computed in a single loop over the entries $\chi(E_{i-1}, \ldots)$.
Then, for every entry $\chi(E_i,\phi,S)$ of the new table, we look in $T$ through all signatures that would be compatible with $(\phi,S)$ and keep the minimum value (and add $Y$ to it, by Lemma~\ref{lem:parent}). The reason we may have to look at several signatures is the set $S$; for simplicity, we may simply loop over all sets $S' \subseteq S$ such that $S' \cup \{j \in [k]: \phi(a,j)>0, \text{ some } a \in E_i\}=S$.
Note that the size of the intermediate table $T$ is immaterial; the time taken consists of first one loop through $\chi(E_{i-1},\ldots)$, then
$2^k$ queries to $T$ for each entry in $\chi(E_i, \ldots)$. Thus, the entries $\chi(E_i, \ldots)$ can all be computed in total time $O^*((c2^k)^p4^k)$.
 As $E_n = \emptyset$ there is only one function $\phi: E_n \times [k] \rightarrow [b,c]$. By Lemma \ref{lem:root}, $\chi(E_n, \phi, [k])$ is the minimum total weight of a solution for $k[b,c]$-DCPP, and $\infty$ if there is no such solution. Thus to solve $k[b,c]$-DCPP it suffices to check the value of $\chi(E_n, \phi, [k])$.
%

 Thus the algorithm finds the value $\weight$ in time $O^*((c2^k)^p4^k)$.

The algorithm can easily be made constructive using the method of Lemma \ref{lem:parent}. For each arc bag $E_i, \arcfunc: E_i\times [k] \rightarrow [0,c], S \subseteq [k]$, in addition to calculating the value $\chi(E_i, \phi, S) = \rho$, we also calculate a witness for $\chi(E_i, \phi, S) \le \rho$, in the cases where $\rho \neq \infty$. Just as we can calculate the values of all $\chi(E_{i}, \cdot, \cdot)$ given the values of all $\chi(E_{i-1}, \cdot, \cdot)$, we may construct witnesses for all $\chi(E_{i}, \cdot, \cdot)$ given witnesses for all $\chi(E_{i-1}, \cdot, \cdot)$, using an intermediate table $T$ as before. (Note that $(A_1, \dots, A_k)$, where each $A_i = \emptyset$, is a witness for $\chi(E_0, \phi, \emptyset) = 0$, where $\phi$ is the empty function. This gives us the base case in our construction of witnesses.)
Given a witness for $\chi(E_n, \phi, [k])$,
Lemma \ref{lem:root} shows how to construct a solution to $k[b,c]$-DCPP on $G$ from this witness.\qed
\end{proof}

\section{Proofs of Theorems \ref{thm:main} and \ref{thm:main2}}\label{sec:main}


\medskip

\noindent{\bf Theorem} \ref{thm:main2}
{\em The $k$-ADCP-{\sc Euler}  is fixed-parameter tractable.}
\begin{proof}
Let $D$ be an Euler digraph. We may assume that $D$ has no vertex of out-degree at least $k$ as otherwise we are done by Lemma \ref{lem:Edeg}.
By Lemma \ref{lem:ptww}, for $D$ we can either obtain $k$ arc-disjoint cycles or a vertex ordering $\nu$ of cutwidth at most $2g(k)$ for some function $g:\ \mathbb{N} \rightarrow \mathbb{N}$.
Note that $D$ is a positive instance of the $k$-ADCP-{\sc Euler}  if and only if $(D,k)$ has a finite solution for $k[0,1]$-DWCP (as every closed walk contains a cycle).
It remains to observe that the algorithm of Theorem \ref{thm:dpa} for the $k[0,1]$-DWCP is fixed-parameter when the out-degree of every vertex of $D$ is upper-bounded by $k$ and the cutwidth of $\nu$ is bounded by a function of $k$.\qed
\end{proof}

\medskip

\noindent{\bf Theorem} \ref{thm:main}
{\em The $k$-DCPP admits a fixed-parameter algorithm.}

\begin{proof}
Let $G=(V,A)$ be a digraph and let $T$ be an optimal solution of DCPP on $G$. 
If  we get a collection $\cal C$ of $k$ arc-disjoint cycles in $G_T$,
then using $\cal C$, by Lemma \ref{lem:kdisjointcycles1}, we can solve the $k$-DCPP on $G$ in (additional) polynomial time. Otherwise, by lemma \ref{lem:ptww}, we have a vertex ordering of $G_T$ of cutwidth bounded by a function of $k$.
We may assume that every vertex
of $G_T$ is of out-degree at most $k-1$
(otherwise by Lemma \ref{lem:Edeg}, $G_T$ has a collection of $k$ arc-disjoint cycles). Since every vertex of $G_T$ is of out-degree at most $k-1$, the multiplicity of $G_T$ is at most $k-1.$
Now Lemma \ref{lem:multiplicity} implies that there is an optimal solution $W$ for the $k$-DCPP on $G$ such that the multiplicity of $G_W$ is at most $k$.
Thus, we may treat the $k$-DCPP on $G$ as an instance $(G,k)$ of $k[1,k]$-DWCP. It remains to observe that the algorithm of Theorem \ref{thm:dpa} to solve the $k[1,k]$-DWCP on $G$ will be fixed-parameter.\qed
\end{proof}

\section{Discussions}\label{sec:dis}

Our algorithms for solving both $k$-DCPP and $k$-ADCP on Euler digraphs have very large running time bounds, mainly because the bound $f(h^{-1}(k))$ on the size of feedback arc set is very large. Function $f(k)$ obtained in \cite{ReRoSeTh1996} is a multiply iterated exponential, where the number of iterations is also a multiply iterated exponential and, as a result, $h^{-1}(k)$ 
grows very quickly.
So obtaining a significantly smaller upper bound for $f(k)$ on Euler digraphs would significantly reduce $h^{-1}(k)$ as well and is of certain interest in itself.
 In particular, is it true that $f(k)=O(k^{O(1)})$ for Euler digraphs? Note that for planar digraphs, $f(k)=k$ \cite[Corollary 15.3.10]{BaGu2009}
 and  Seymour \cite{Sey1996} proved the same result for a wide family of Euler digraphs.
 It would also be interesting to check whether the $k$-DCPP or $k$-ADCP admits a polynomial-size kernel.

 Cechl\'{a}rov\'{a} and Schlotter \cite{CeSc2010} introduced the following
somewhat related problem in the context of housing markets:  can we delete
at most $k$ arcs in a given digraph such that each strongly connected
component of the resulting digraph is Euler? They asked for the
parameterized complexity of this problem, where $k$ is the parameter.
Crowston {\em et al.} \cite{CrGuJoYe2012} showed that the problem
restricted to tournaments is fixed-parameter tractable, but in general the
complexity still remains an open question. See also the recent paper
\cite{CyMaPiPiSc2014} for other related problems.

\medskip

\medskip

\noindent{\bf Acknowledgement} Research of GG was supported by Royal Society Wolfson Research Merit Award.

\newpage


\vspace{1cm}

{\large \bf Appendix: Proof of Lemma \ref{lem:parent}}

\vspace{0.3cm}

\noindent{\bf Lemma \ref{lem:parent}}
 {\em  Consider an arc bag  $E_{i}$, for $i \ge 1$.
   Let $E_i^* = E_{i} \setminus E_{i-1}$.
   For any  $\arcfunc :E_{i} \times [k] \rightarrow [0,c]$ and $S\subseteq [k]$,
   let $Y = \sum_{j\in S}\sum_{a \in E_i^*}\arcfunc(a,j) \cdot \omega(a)$.

   If there exists $a \in E_i$ such that  $\sum_{j \in [k]}\arcfunc(a,j) < b$ or  $\sum_{j \in [k]}\arcfunc(a,j)>c$, then
   $\chi(E_{i},\arcfunc,S) = \infty$.

  Otherwise, the following recursion holds:

  \[\chi(E_{i},\arcfunc,S) =  Y + \min_{\arcfunc',S'}\chi(E_{i-1},\arcfunc',S')\]

   where the minimum is taken over all $\arcfunc' :E_{i-1} \times [k] \rightarrow [0,c]$, and $S'\subseteq [k]$ satisfying the following conditions:
   \begin{itemize}
    \item For all $a \in E_{i}\cap E_{i-1}$ and all $j \in [k]$, $\arcfunc'(a,j) = \arcfunc(a,j)$;
   \item For all $j \in [k]$,
   \begin{eqnarray*}
      \sum_{a \in A^+(v_i) \cap E_{i-1}}\phi'(a,j) + \sum_{a \in A^+(v_i) \cap E_i} \phi(a,j) \\
    =   \sum_{a \in A^-(v_i) \cap E_{i-1}}\phi'(a,j) + \sum_{a \in A^-(v_i)\cap E_i}\phi(a,j).
   \end{eqnarray*}

    \item $S = S' \cup \{j \in [k]: \sum_{a \in E_i^* }\arcfunc(a,j) > 0\}$.
   \end{itemize}
    If there are no $\arcfunc',S'$ satisfying these conditions, then $\chi(E_{i},\arcfunc,S) = \infty$.

 Furthermore, if there exist $\arcfunc',S'$ satisfying the above conditions and we are given a witness $(A_1', \dots, A_k')$ for $\chi(E_{i-1}, \arcfunc', S') \le \rho'$, then we can construct a witness for $\chi(E_{i},\arcfunc,S) \le Y + \rho'$ in polynomial time.
}

\medskip 

\noindent{\bf Proof}
We will prove the last claim of the lemma first. Suppose we are given a witness $(A_1', \dots, A_k')$ for $\chi(E_{i-1}, \arcfunc', S') \le \rho'$.
For each $j \in [k]$, let $A_j$ be the multiset $A_j'$ together with $\arcfunc(a,j)$ copies of each arc in $E_i^*$.
 We now show that $(A_1, \dots, A_k)$ is a witness for
 $\chi(E_{i},\arcfunc,S) =  Y + \rho'$.

 By construction of $A_j$, definition of $A_j'$ and the fact that $\arcfunc'(a,j) = \arcfunc(a,j)$ for all $a \in E_{i} \cap E_{i-1}$, $j \in [k]$, we have that for all $a \in E_{i}$ and $j \in [k]$, $A_j$ contains exactly $\arcfunc(a,j)$ copies of $a$, satisfying Condition \ref{con:ffunction} of $\chi(E_{i},\arcfunc,S) \le Y + \rho'$.

  By definition of $A_j'$ and the fact that $b \le \sum_{j \in [k]}\arcfunc(a,j) \le c$ for each $a \in E_i^*$, we have that every arc appears at least $b$ times and at most $c$ times in $A_1 \cup \dots \cup A_k$, satisfying Condition \ref{con:edgecovering}.

 Observe that $E_i^*$ consists of all arcs of the form $v_iv_h$ or $v_hv_i$ for $h > i$.
It follows by construction that for any $h < i$ and $j \in [k]$, $A^+(v_h,A_j)=A^+(v_h,A'_j)$ and $A^-(v_h,A_j)=A^-(v_h,A'_j)$.
  Then as $|A^+(v_h, A_j')| = |A^-(v_h, A_j')|$ for all $h < i$, we have that $|A^+(v_h, A_j)| = |A^-(v_h,A_j)|$ for all $h < i$.
 As every arc incident with $v_i$ is in exactly one of $E_{i-1}$ or $E_i$, we have that for all $j \in [k]$,  $|A^+(v_i, A_j)|  = \sum_{a \in A^+(v_i) \cap E_{i-1}}\phi'(a,j) + \sum_{a \in A^+(v_i) \cap E_i} \phi(a,j)$, and similarly $|A^-(v_i, A_j)| = \sum_{a \in A^-(v_i) \cap E_{i-1}}\phi'(a,j) + \sum_{a \in A^-(v_i)\cap E_i}\phi(a,j)$. It follows by the second condition of the lemma that $|A^+(v_i, A_j)| = |A^-(v_i, A_j)|$.
Therefore $|A^+(v_h, A_j)| = |A^-(v_h, A_j)|$ for all $h \le i$, satisfying Condition \ref{con:vertexbalance}.

 By the fact that $S = S' \cup \{j \in [k]: \sum_{a \in E_{i} \setminus E_{i-1} }\arcfunc(a,j) > 0\}$, definition of $(A_1', \dots, A_k')$ and construction of $(A_1, \dots, A_k)$, we have that $S = S' \cup \{j \in [k]: A_j \setminus A_j' \neq \emptyset\} = \{j \in [k]: A_j \neq \emptyset\}$. This satisfies Condition \ref{con:sset}.

 Finally, by construction of $\{A_1, \dots ,A_k\}$ and $\chi(E_{i-1},\arcfunc',S')$, we have that
 $\chi(E_{i},\arcfunc,S)=\sum_{j \in [k]} \sum_{a \in A_j} \omega(a) = Y + \sum_{j \in [k]} \sum_{a \in A_j'} \omega(a) = Y + \chi(E_{i-1},\arcfunc',S')$, satisfying Condition \ref{con:wweight}.

 Thus, we have that $(A_1, \dots, A_k)$ is a witness for $\chi(E_{i},\arcfunc,S) \le  Y + \rho'$.

 We now prove the other claims of the lemma.
  If there exists $a \in E_i$ such that  $\sum_{j \in [k]}\arcfunc(a,j) < b$ or  $\sum_{j \in [k]}\arcfunc(a,j)>c$, then any arc multisets $A_1, \dots, A_k$ that satisfy Condition \ref{con:ffunction} of $\chi(E_{i},\arcfunc,S)$ will falsify Condition \ref{con:edgecovering}, and so  $\chi(E_{i},\arcfunc,S) = \infty$.
  So now assume that $b \le \sum_{j \in [k]}\arcfunc(a,j) \le c$ for every $a \in E_i$.

  Let $\arcfunc' :E_{i-1} \times [k] \rightarrow [b,c]$, $S'\subseteq [k]$ be such that the conditions of the lemma are satisfied and $\chi(E_{i-1},\arcfunc',S')$ is minimised.
  If $\chi(E_{i-1},\arcfunc',S') = \infty$ then trivially $\chi(E_{i},\arcfunc,S) \le  Y + \chi(E_{i-1},\arcfunc',S')$.
  Otherwise, $\chi(E_{i-1},\arcfunc',S') = \rho' \neq \infty$ and so there exists a witness for $\chi(E_{i-1},\arcfunc',S') \le \rho'$
  Then by the argument above, there exists a witness for $\chi(E_{i},\arcfunc,S) \le  Y + \chi(E_{i-1},\arcfunc',S')$.
  In either case $\chi(E_{i},\arcfunc,S) \le  Y + \chi(E_{i-1},\arcfunc',S')$.

  It remains to show that if $\chi(E_{i},\arcfunc,S) \neq \infty$, then there exist $\arcfunc',S'$ such that $\chi(E_{i},\arcfunc,S) =  Y + \chi(E_{i-1}\arcfunc',S')$.

  Suppose that  $\chi(E_{i},\arcfunc,S) = \weight \neq \infty$.
  Let $(A_1, \dots, A_k)$ be a witness for  $\chi(E_{i},\arcfunc,S) = \weight$.
  Then for each $j\in [k]$, let $A'_j$ be the multiset of arcs
from $A_j$ not incident to $v_i$ and let $A^*_j$ be the multiset of arcs  from
$A_j$ incident to $v_i$. For a multiset $M$ of arcs from $G$, let
$\omega(M)=\sum_{a\in M}\omega(a)$, where each arc $a$ is taken in the sum
as many times as it has copies in $M.$ Observe that $Y=\sum_{j\in
[k]}\omega(A^*_j)$.
  Let $Z = \sum_{j \in [k]}\omega(A_j')$; then $\chi(E_{i},\arcfunc,S) =  Y +Z$.

  Let $\arcfunc' :E_{i-1} \times [k] \rightarrow [b,c]$ be the function such that $\arcfunc'(a,j)$ is the number of copies of $a$ in $A_j'$, for each $a \in E_{i-1}, j \in [k]$.
  Finally let $S' = \{j \in [k]: A_j' \neq \emptyset\}$.

  As $\gamma(i) \setminus \gamma(i-1)$ contains no arcs incident to $v_h$ for any $h < i$, we have that for any $h < i$, $|A(v) \cap A_j'| = |A(v) \cap A_j|$ for each $j \in [k]$. Therefore $(A_1', \dots, A_k')$ satisfies Conditon \ref{con:vertexbalance} of a witness for  $\chi(E_{i-1},\arcfunc',S') = Z$. It is easy to see that $(A_1', \dots, A_k')$ satisfies the other conditions for a witness for $\chi(E_{i-1},\arcfunc',S') = Z$, from which it follows that $\chi(E_{i},\arcfunc,S) =  Y +\chi(E_{i-1},\arcfunc',S')$.\qed
%

\end{document}